\documentclass{article}

\usepackage{graphicx}
\usepackage{amsmath,amssymb,latexsym,epsfig,epstopdf,amsthm}
\usepackage{bbm}
\usepackage{a4wide}
\usepackage{hyperref}

\newtheorem{lemma}{Lemma}
\newtheorem{theorem}{Theorem}
\newtheorem{proposition}{Proposition}
 
\def\bp{\noindent{\it Proof.}\ } 
\def\ep{\hfill $\Box$}  

\def\sumj{\sum_{j=1}^J} 

\def\sumk{\sum_{k=1}^K} 
\def\prodj{\prod_{j=1}^J} 
 
\def\sumy{\sum_{y\in \Y}} 

\def\D{{\cal D}}
\def\X{{\cal X}}
\def\Y{{\cal Y}}
\def\Z{{\cal Z}}

\def\N{{\mathbb N}}
\def\R{{\mathbb R}}

\newcommand{\diff}{\mathop{}\mathopen{}\mathrm{d}}
\newcommand{\ind}{\ensuremath{\mathbbm{1}}}

\title{Performance of CSMA in Multi-Channel Wireless Networks\thanks{A preliminary version of this paper was presented at   CISS 2011 \cite{BF-ciss11}.}}
\author{Thomas Bonald\thanks{T.~Bonald is with the Department
of Computer Science and Networking, Telecom ParisTech, Paris, France. Email: 
thomas.bonald@telecom-paristech.fr} and Mathieu Feuillet\thanks{M.~Feuillet is with INRIA, Paris-Rocquencourt, France. Email: 
mathieu.feuillet@inria.fr}}

\date{\today}

\begin{document}

\maketitle

\begin{abstract}
We analyze the performance of CSMA  in multi-channel wireless networks, accounting for the random nature of traffic. Specifically, we assess the ability of CSMA to fully utilize the radio resources and in turn to  stabilize the network in a dynamic setting with flow arrivals and departures. We prove that CSMA is optimal in ad-hoc mode but not in infrastructure mode, when all data flows originate from or are destined to some access points, due to the inherent bias of CSMA against downlink traffic. We propose a slight modification of CSMA, that we refer to as \textit{flow-aware} CSMA, which corrects this bias and makes the algorithm optimal in all cases.
The analysis is based on some time-scale separation assumption which is proved valid in the limit of large flow sizes.
\\

\noindent
{\bf Keywords:} Wireless network, interference graph, CSMA, flow-level dynamics, time-scale separation, stability.
\end{abstract}

\section{Introduction}
The CSMA (Carrier Sense Multiple Access) algorithm is a  key component of  IEEE 802.11 networks. 
While it  proves successful in sharing a single  radio channel between a limited number of stations, 
its efficiency  is questionable in more involved environments with multiple radio channels and a large number of 
stations having different  interference constraints. In this paper, we analyse the ability of CSMA to fully utilize the radio resources in such environments, in both ad-hoc and infrastructure modes, accounting for the random nature of traffic. Specifically, each station attempts to access a randomly chosen radio channel after some random backoff time and transmits a packet over this channel if it is sensed idle. We study the random variations of the number of active wireless links induced by this random access algorithm and the random activity of users. In particular, we analyse the ergodicity of the associated Markov process, which characterizes the ability of CSMA to stabilize the network.

It turns out that, while CSMA is always efficient in ad-hoc mode, in the sense that the network is stable whenever possible, it is generally inefficient in infrastructure mode, when all data flows originate from or are destined to some  finite set of access points. 
This is due to the inherent bias of CSMA against {\it downlink} traffic, from the access points to the stations: each access point attempts to access the radio channels with the same rate, independently of the number of active downlink flows at this access point. We prove that a slight modification of CSMA, which
consists  in running one instance of CSMA per flow at each access point, corrects this bias and makes the algorithm optimal. We refer to this algorithm, introduced in \cite{BF-10},   as \textit{flow-aware} CSMA.

The rest of the paper is organized as follows. We present some related work in the next section. The network model in ad-hoc mode is described in section \ref{sec:model}. Sections \ref{sec:packet-level} and \ref{sec:flow-level} are devoted to the packet- and flow-level dynamics, respectively, assuming time-scale separation. The main result of the paper, given in Theorem \ref{optim}, shows in particular the  optimality of CSMA in ad-hoc mode. The validity of the time-scale separation assumption is discussed in section \ref{sec:timescale}. The infrastructure mode is considered in section \ref{sec:downlink}, where we prove the suboptimality of standard CSMA and the optimality of flow-aware CSMA. Section \ref{sec:conc} concludes the paper.

\section{Related work}
The present work is related to the problem of optimal scheduling in wireless networks. While a centralized solution is known since the seminal work of 
Tassiulas and Ephremides, who proved  in  \cite{tassiulas-92} the optimality of  the \textit{maximum weight} policy, no distributed solution was known until the recent works of Jiang, Ni, Shah and Walrand   \cite{shah-09,walrand-08,raja-09}. These authors considered a simple CSMA algorithm whereby the    attempt rate  of each station depends either  on  the  number of queued packets or on 
some local estimates  of the arrival rate and the service rate of packets  at the station. Similar ideas are used by Ni, Tan and Srikant in   \cite{srikant-10}. 
The proof of optimality relies on the fact that these adaptive versions of CSMA achieve the
  maximum weight scheduling, under some technical assumptions related to the speed of convergence of the algorithm. In practice,
  the algorithm must indeed be carefully designed so as to enforce the time-scale separation, as shown   for instance  in the recent paper of  Prouti\`ere, Yi, Lan and Chiang   \cite{proutiere-10}.

All these papers focus on the packet-level dynamics, assuming packets are generated by some fixed number of flows. The flow-level dynamics are ignored, whereas they are 
 known to be critical, see for instance  \cite{Barakat03,ben01,BK00,MR00} in the context of wireline networks. As in our previous paper \cite{BF-10}, we consider both the packet- and flow-level dynamics, under the usual  assumption that the former are much faster than the latter. 
 Specifically, we extend the results 
 of \cite{BF-10} to multi-channel networks in both ad-hoc and infrastructure modes and discuss the validity of the time-scale separation assumption.
 
 Surprisingly, little attention  has so far been paid   to multi-channel networks. A notable exception is the adaptive, multi-channel version of CSMA introduced in  \cite{proutiere-10}, which is shown to maximize the network utility when combined with some appropriate virtual queue  mechanism. We here prove the optimality of CSMA in the sense of flow-level stability for a  very general model where the interference constraints may depend on the considered channel and each transmitter may only use a subset of the channels. 
 Specifically, we show that it is sufficient for  each transmitter to probe one of its channels at random, without any further information on the network state. 
  
 Another salient feature of this paper is the observation of the key difference between the ad-hoc and infrastructure modes.  In the former, the number of transmitters grows with the congestion, which increases the channel attempt rate and in turn  stabilizes the network. This is not the case of the latter since the channel access opportunities of each access point must be shared by all downlink flows at this access point. This inherent bias of CSMA against downlink traffic is well known, see e.g.~\cite{heusse-03,kim05}, and can be easily corrected by letting  the attempt rate of each access point depend on the number of downlink flows, a scheme we refer to as \textit{flow-aware} CSMA \cite{BF-10}. The algorithm is then optimal.

\section{Model}
\label{sec:model}

  \subsection{A multi-channel wireless network}

The network consists of a random, dynamic set of wireless links 
in ad-hoc mode (there is no access point at this stage). These links must share 
some finite number $J$ of non-interfering radio channels.   Each  link consists of a  transmitter-receiver pair; the
transmitter  is able to  use at most one  radio channel at a time. 
We group links into a finite number of  $K$ {\it classes}, as illustrated  by Figure \ref{fig:first-example}. All links within the same class have the same  radio conditions, the same interference constraints and the same CSMA parameters.
We denote by $x_k$ the number of class-$k$ links and by $x$ the corresponding vector, which we refer to as the \textit{network state}. Two links within the same class cannot be simultaneously active on the same channel.  An active  class-$k$ link on channel $j$ transmits data at the physical rate $\varphi_k$ bit/s, independently of $j$.
We say that class $k$ is active on channel $j$ if there is an active  class-$k$ link on channel $j$.  

\begin{figure}[h]
\centering
\begin{minipage}[c]{0.5\linewidth}
\includegraphics[width=\linewidth]{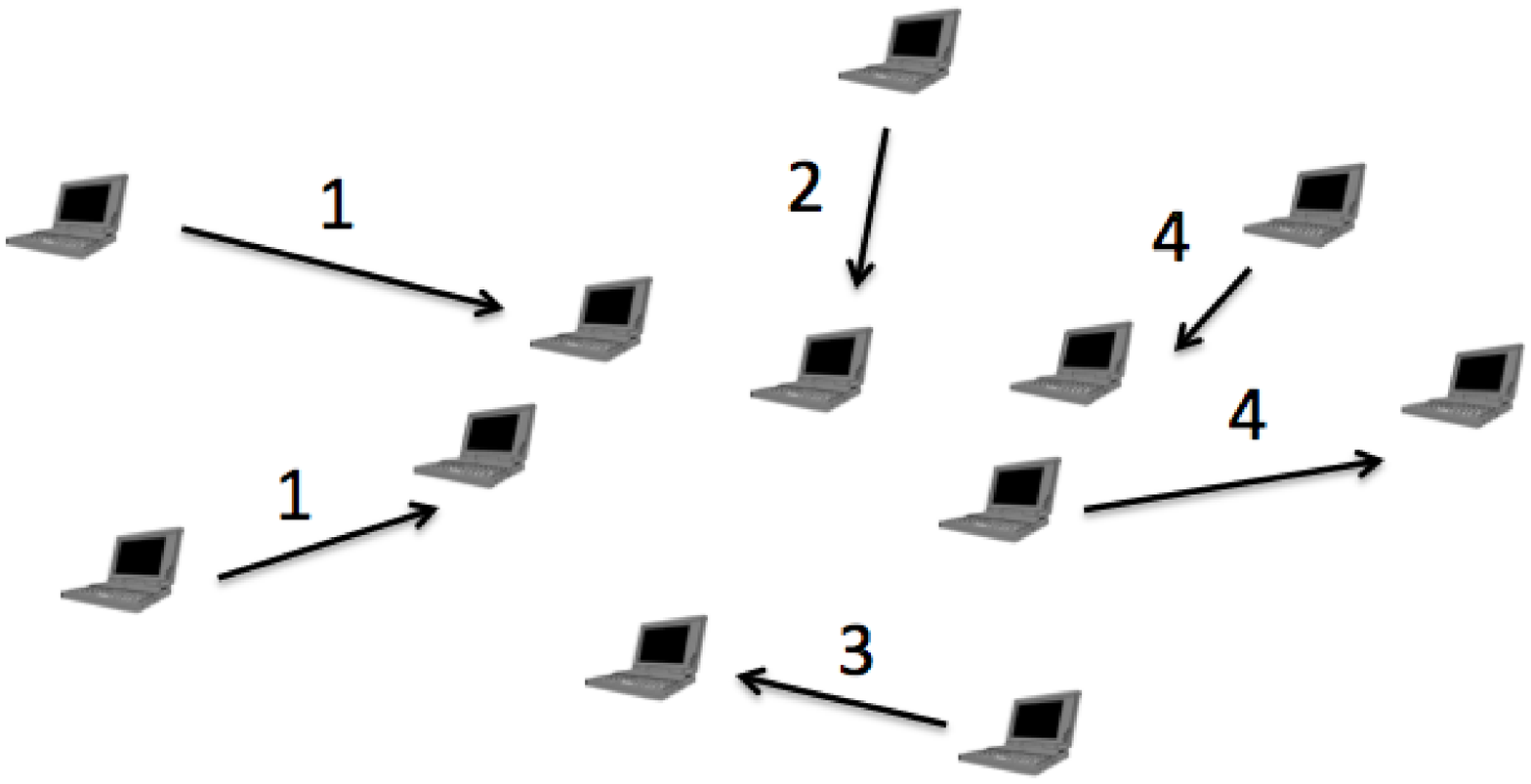}
\end{minipage}
\hspace{0.05\linewidth}
\begin{minipage}[c]{0.25\linewidth}
\includegraphics[width=\linewidth]{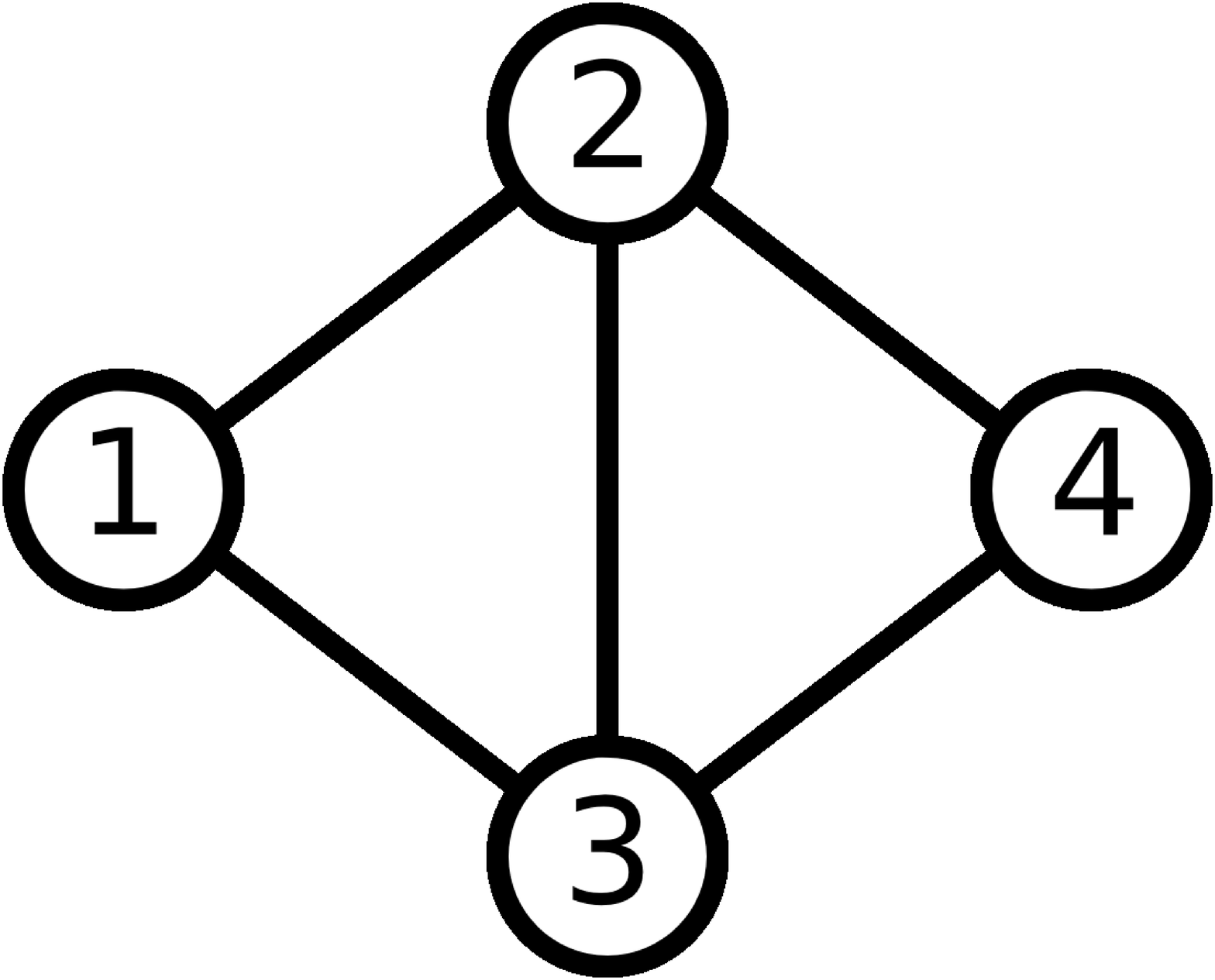}
\end{minipage}
\caption{An ad-hoc wireless network with 4 classes of links and its interference graph.}
\label{fig:first-example}
\end{figure}

Each channel $j$ is
associated with some \textit{conflict graph}  $G_j=(V_j,E_j)$, where $V_j\subset\{1,\ldots,K\}$ is the
set of classes that are able to transmit on channel $j$ and $E_j$ is
the set of edges, each representing a conflict. Specifically, two classes $k,l\in V_j$ can be
simultaneously active on channel $j$ if and only if they do not conflict with each other, that
is if $(k,l)\not \in E_j$. The $J$ conflict graphs are typically the same but
could differ due to different radio propagation environments on the $J$ channels,
 or to different transmission capabilities of the $K$ classes.

 \subsection{Feasible schedules}
\label{sec:fs}

We refer to a schedule as any vector $y\in \{0,1\}^{K\times J}$, where $y_{kj}=1$ if class $k$ is active on channel $j$. We denote by $y_k$ the number of active class-$k$ links:
$$
y_k= \sum_{j=1}^J y_{kj}.
$$
The schedule is \textit{feasible} if for all $j=1,\ldots,J$,  the active classes on channel $j$ belong to $V_j$ and do not conflict with each other, that is 
$y_{kj}y_{lj}=0$ for all $(k,l)\in E_j$. Moreover, we must have:
\begin{equation}\label{eq:yx}
\forall k=1,\ldots,K,\quad y_k\le x_k.
\end{equation}
We denote by  $\Y(x)$ the set of feasible schedules. Note that if $x_k\ge J$ for all $k=1,\ldots,K$, the constraint \eqref{eq:yx} is no longer limiting (since the number of active class-$k$ links is limited by the number of radio channels $J$) and the set of feasible schedules becomes independent of the network state. We denote by $\Y$ the corresponding set, which is the union of $\Y(x)$ over all network states $x$.

 \subsection{Capacity region}
\label{capa}
Assume that each feasible schedule $y$ is selected with probability $\pi(y)$, with $\sumy \pi(y)=1$. 
The mean throughput of class $k$  is then given by:
\begin{equation}\label{eq:mu}
\phi_k=\varphi_{k} \sum_{y\in \Y} y_k \pi(y).
\end{equation}
Let  $\phi$ be the corresponding throughput vector.
We refer to the \textit{capacity region} as the set of vectors $\phi$ generated by all probability measures $\pi(y)$, $y\in \Y$.
Note that the capacity region depends both on the physical rates and on the interference constraints of all wireless links.

\section{Packet-level dynamics}
\label{sec:packet-level}
We first analyze the packet-level dynamics induced by CSMA for a static network state $x$. The flow-level dynamics that make $x$ vary are introduced in section \ref{sec:flow-level}.

\subsection{Random access}
\label{sec:ra}

We consider the standard CSMA algorithm where each transmitter  waits for a period
of random duration referred to as the \textit{backoff time} before each transmission
attempt. At each attempt, the transmitter chooses a radio channel at random and probes it.
If the radio channel is sensed idle (in the sense that no conflicting link is active),
a packet  is transmitted (we neglect tho channel after some random backoff time and transmits a packet over this channel if it is sensed idle. We study the random variations oe collisions); otherwise, the transmitter waits for a new backoff time
before the next attempt. 

Packets have random sizes of unit mean   and are transmitted at the physical rate $\varphi_k$ on class-$k$ links; the backoff times
 of class-$k$ transmitters are random with mean $1/\nu_k$ , where $\nu_k>0$ is the corresponding attempt rate. We denote by $\alpha_k=\nu_k/\varphi_k$
the ratio  of the mean packet transmission time to the mean backoff time of  class-$k$ links. 
Channel $j$ is chosen with probability $\beta_{kj}$, with
$\sum_{j=1}^J \beta_{kj} = 1$ and $\beta_{kj} > 0$ if and only if $k\in V_j$, so that all accessible channels are attempted with positive probability.

\subsection{Stationary distribution}

Let $Y(t)$ be the schedule selected by the  above random access algorithm at time $t$. We look for the stationary distribution of $Y(t)$, which we denote by $\pi(x,y)$ to highlight the fact that it depends on the network state $x$. We have:

\begin{proposition}\label{prop1}
If both the packet sizes and the
backoff times  have exponential distributions, then $Y(t)$  is a reversible Markov process, with stationary measure:
\begin{equation}\label{eq:wix}
w(x,y)= \prod_{k:x_k>0} {x_k!\over (x_k-y_k)!}\alpha_k^{y_k}\prodj {\beta_{kj}^{y_{kj}}},\quad y\in \Y(x).
\end{equation}
\end{proposition}
\bp
Let $e_{kj}$ be the unit vector on component $k,j$ on  $\{0,1\}^{K\times J}$. 
The Markov process $Y(t)$ jumps from state $y$ to state $y+e_{kj}$ with rate $(x_k-y_k) \nu_k \beta_{kj}$ (since all idle links attempt to access the channel) and from  state $y+e_{kj}$ to state $y$ with rate $\varphi_k$ (since all class-$k$ links have the same physical rate $\varphi_k$, independently of the used channel), for any state $y$ such that $y+e_{kj}\in \Y(x)$.
The proof then follows from the local balance equations:
$$
w(x,y) (x_k-y_k)\nu_k\beta_{kj}=w(x,y+e_{kj}){\varphi_k}.
$$
\ep

The stationary distribution $\pi(x,y)$ follows from the normalization of  the stationary measure $w(x,y)$  over all $y\in \Y(x)$. We deduce the mean throughput of class $k$ in state $x$:
\begin{equation}\label{eq:phix}
\phi_k(x)=\varphi_k \sumy y_k \pi(x,y).
\end{equation}
It turns out that, by the insensitivity property of the underlying loss network
\cite{bonald-07}, these expressions are in fact valid for any \textit{phase-type} distributions of packet sizes and backoff times; such distributions are known to form a dense subset within
the set of all distributions with real, non-negative support \cite{serfozo}, so
that the results hold for
virtually any distributions of packet sizes and backoff times. We refer the reader to \cite{VBLP10} for further details on this insensitivity property.

\section{Flow-level dynamics}
\label{sec:flow-level}
We now introduce the flow-level dynamics under the assumption of 
infinitely fast packet-level dynamics; the validity of this 
time-scale separation assumption is discussed in  section \ref{sec:timescale}.

\subsection{Traffic characteristics}

We assume that flows using class-$k$ links are generated according to a Poisson process of intensity $\lambda_k$.  
  Each such flow has an exponential size with mean $\sigma_k$ bits and leaves the network once the corresponding data transfer is completed. 
There is a one-to-one correspondence between flows and links so that both terms are used
 interchangeably in the following. 
  We denote by
$\rho_k=\lambda_k\sigma_k$ the traffic intensity of class $k$ (in bit/s) and by
$\rho$ the corresponding vector. 
 
Under the time-scale separation assumption,  the 
 flow-level dynamics are much slower than the 
 packet-level dynamics   so that, at the  time scale of a flow,  everything happens as if the stationary distribution \eqref{eq:wix} of the packet-level dynamics were reached instantaneously.  
 In particular, the mean throughput of class $k$ is given by \eqref{eq:phix} in state $x$.

\subsection{Stability region}

 Let $X_k(t)$ be the number of class-$k$ flows at time $t$. The   corresponding vector $X(t)$ describes
 the evolution of the network state. 
This
  is a Markov process  with transition rates $\lambda_k$  from state $x$ to state $x+e_k$ and 
$\phi_k(x)/\sigma_k$ from state $x$ to state $x-e_k$ (provided $x_k>0$),
where $e_k$ denotes the unit vector on component $k$.
 
 We say that the network is stable if this Markov process is ergodic. 
Clearly, a necessary condition for stability is that the vector of traffic intensities $\rho$ lies in the capacity region. The following key result of the paper shows that this condition is 
in fact sufficient, up to the critical case where $\rho$ lies on the boundary of the capacity region. In this sense, CSMA is optimal in the considered ad-hoc mode.

\begin{theorem}\label{optim}
The network is stable for all vectors of traffic intensities $\rho$ in the interior of the capacity region.
\end{theorem}

The proof is deferred to the appendix. It is based on the fact that the 
random access algorithm selects schedules in proportion to their weights (\ref{eq:wix}). For large $x$, this is equivalent to selecting schedules in proportion to the following uniform weight, which is independent of the channel probing distribution:
\begin{equation}\label{eq:weight}
u(x,y)=\prod_{k: x_k>0} (x_k{\alpha_k})^{y_k}, \quad y\in \Y(x).
\end{equation}
Defining:
$$
u(x)=\max_{y\in \Y(x)}u(x,y),
$$
the following result, also proved in the appendix, shows that  those schedules of maximum weight are actually selected with probability close to 1:
\begin{lemma}
For any $\epsilon>0$, we have:
$$
\sum_{y\in \Y(x)} \pi(x,y)\log(u(x,y))\ge (1-\epsilon) \log(u(x))
$$
for all states $x$ but some finite number.
\end{lemma}

The result then follows from the stable behavior of maximum weight scheduling, except that the latter is defined over the set of {\it all} feasible schedules. Defining the corresponding weight by:
$$
v(x)=\max_{y\in \Y}u(x,y),
$$
the following result, proved in the appendix, shows that  it  is essentially the same as $u(x)$:

\begin{lemma}
We have:
  $$\sup_{x\in \X}{v(x)\over  u(x)}<\infty.$$
\end{lemma}

The proof of Theorem 1, based on Lemmas 1 and 2, then follows from Foster's criterion.

\section{Time-scale separation}
\label{sec:timescale}
Theorem \ref{optim} is based on the time-scale separation assumption: 
in the 
packet-level model of
 section \ref{sec:packet-level}, packets   ``see''  a fixed number of flows, while in the
flow-level model of section
\ref{sec:flow-level}, 
flows ``see''  the equilibrium  state of packet-level dynamics. In this section, we remove this assumption. Specifically, we prove that when the size of the flows
grows, the model without time-scale separation  converges  to the  model with time-scale separation, which indeed  suggests that CSMA is optimal for sufficiently large flow sizes. We actually conjecture that CSMA is optimal for any flow size, which we prove at the end of the section for  
 a specific class of networks. 

\subsection{Scaling}

As in section \ref{sec:flow-level}, class$-k$ flows are assumed to arrive according to a Poisson
process of intensity $\lambda_k$. The number of packets per class-$k$ flow has a geometric distribution with mean $N\sigma_k$, where $N$ is some positive integer, we refer to as the   scaling parameter.  In particular, each class-$k$ flow terminates with probability $1/(\sigma_k N)$ after each packet transmission. 
Packets are assumed to have an exponential
size with mean $1/N$ bits, so as to keep  the class-$k$ mean flow size constant and equal to $\sigma_k$ bits. In particular,  the corresponding traffic intensity
$\rho_k=\lambda_k\sigma_k$ is independent of $N$. 

The random access algorithm is that described in section \ref{sec:ra}. The only difference is that the attempt rates must be scaled so as to keep the 
 ratio of mean packet transmission time to mean backoff time constant. Thus each class-$k$ link now attempts to access the channels at rate $N\nu_k$.

\subsection{Asymptotic time-scale separation}

The state of the network  is now described  by the couple 
 $(X^N(t),Y^N(t))$, where $X^N(t)$ gives the number of flows of each class at time $t$ and $Y^N(t)$ the schedule that is selected at time $t$. This is a  Markov process with transition rates $\lambda_k$ from 
  state $(x,y)$ to state $(x+e_k,y)$ (class-$k$ flow arrival), $N(x_k-y_k)\nu_k \beta_{kj}$ from state $(x,y)$ to state $(x,y+e_{kj})$ (access to channel $j$ by a class-$k$ flow), $Ny_{kj} \varphi_k(1 - {1}/({\sigma_kN}))$ from  state $(x,y)$ to state $(x,y-e_{kj})$ (packet transmission of a class-$k$ flow over channel $j$, without flow completion), $y_{kj} {\varphi_k}/{\sigma_k}$ from  state $(x,y)$ to state $(x-e_k,y-e_{kj})$ (packet transmission of a class-$k$ flow over channel $j$, with flow completion).
 
When $N$ grows, the packet-level dynamics, represented by  $Y^N(t)$,
are accelerated with respect  to the flow-level dynamics, represented by  $X^N(t)$. The following result, proved in  the appendix, shows that there is indeed time-scale separation between the packet level and the flow level in the limit. We assume that $X^N(0) = X(0)$ for all $N\ge 1$.

\begin{theorem}
When $N\to\infty$, the stochastic process $X^N(t)$ converges in distribution to the Markov process $X(t)$, which describes the network state under the time-scale
separation assumption.
\label{theo:time-scale}
\end{theorem}

\subsection{Stability of some class of networks}

Theorems \ref{optim} and \ref{theo:time-scale} suggest that CSMA is optimal for sufficiently large flow sizes. We conjecture that CSMA is actually optimal for any flow size, in the sense that the Markov process  $(X^N(t),Y^N(t))$ is ergodic for any scaling parameter $N\ge 1$ provided the 
vector of
traffic intensities $\rho$ lies in the interior of the capacity region.
To support this conjecture, consider the following class of networks. We assume that all links have access to the $J$ channels.
The interference graph is the same
on all channels and given by some
$L$-partite
graph, i.e. there exists a partition $\{C_1, \dots,C_L\}$ of $\{1,\dots,K\}$ such that
two classes in $C_l$ do not interfere with each other but a class in $C_l$ does interfere with all classes 
in $\{1,\ldots,K\}\setminus C_l$. Examples of $L$-partite graphs are given in figure \ref{fig:partite}.
The following result, proved in the appendix, shows that CSMA is optimal independently of the scaling parameter $N$:

\begin{proposition}\label{bipart}
Any network with a $L$-partite interference graph is stable for all vectors of
traffic intensities $\rho$ in the interior of the capacity region.
\end{proposition}

\begin{figure}[h!]
\centering
\begin{minipage}[b]{0.25\linewidth}
\centering
\includegraphics[width=\linewidth]{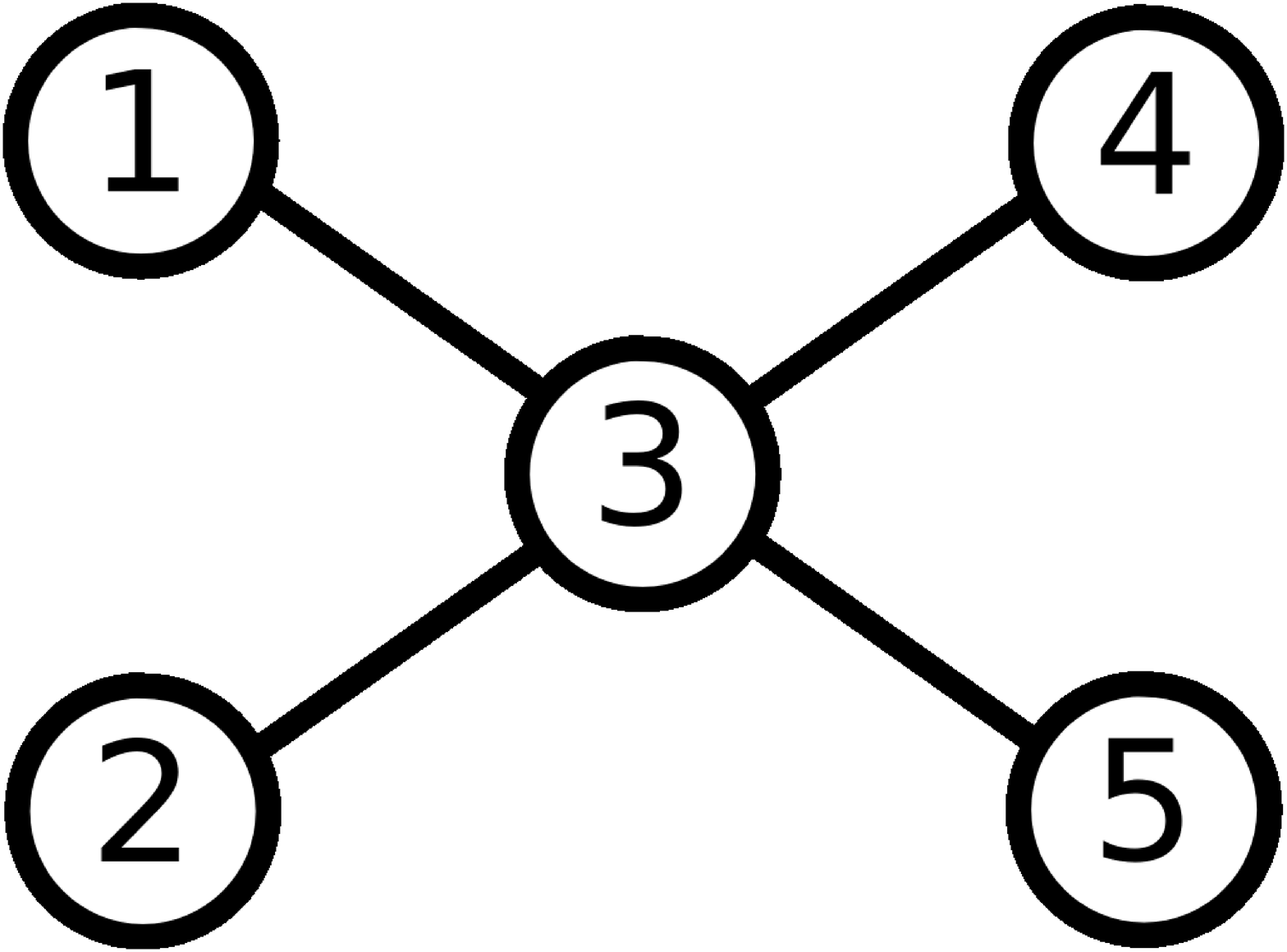}
\vspace{0.1cm}
\begin{align*}
C_1&=\{3\},\\ C_2&=\{1,2,4,5\}.
\end{align*}
(a)
\end{minipage}
\hspace{0.05\linewidth}
\begin{minipage}[b]{0.25\linewidth}
\centering
\includegraphics[width=0.84\linewidth]{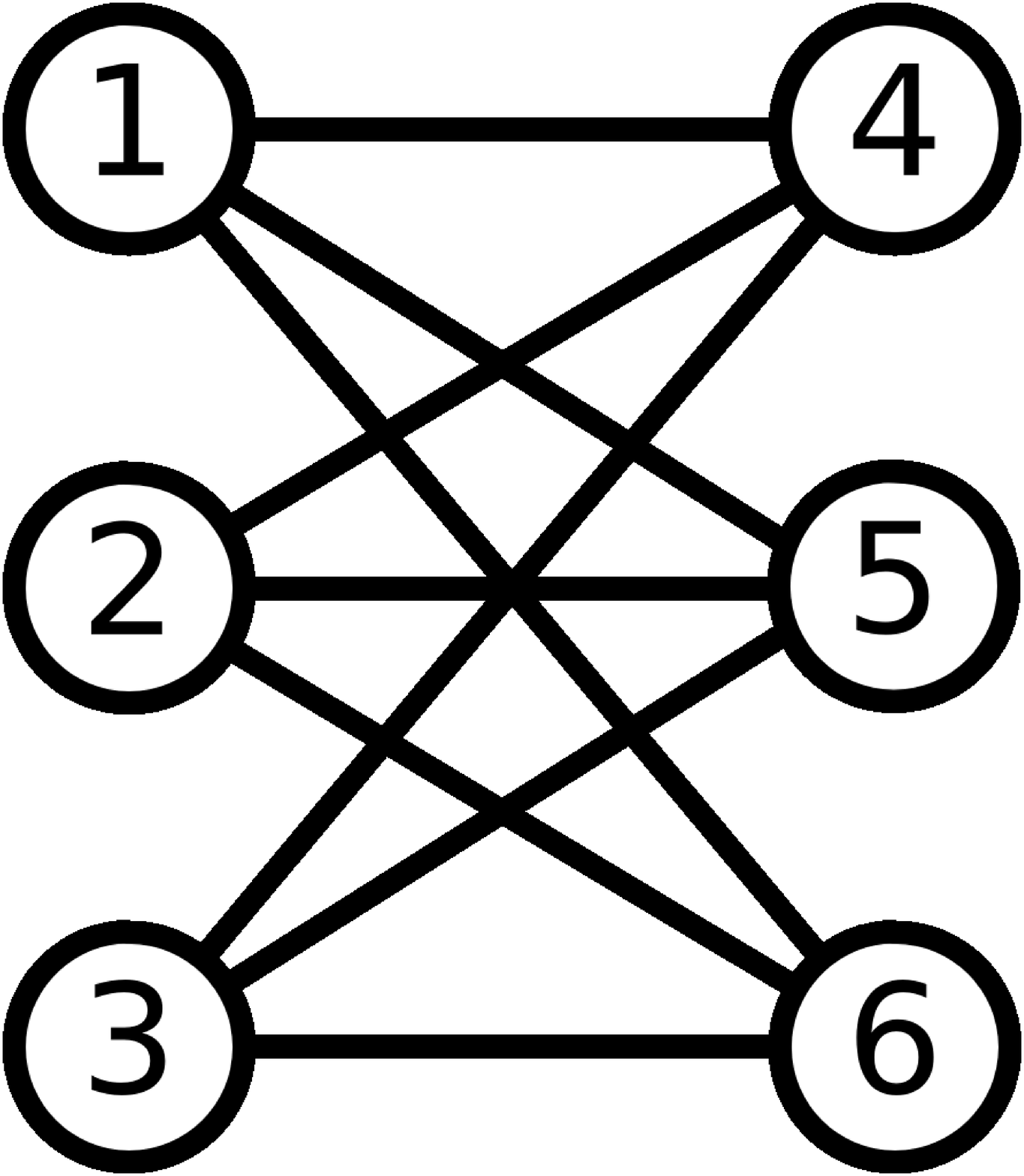}
\vspace{0.4cm}
\begin{align*}
C_1&=\{1,2,3\},\\ C_2&=\{4,5,6\}.
\end{align*}
(b)
\end{minipage}
\hspace{0.05\linewidth}
\begin{minipage}[b]{0.25\linewidth}
\centering
\includegraphics[width=0.81\linewidth]{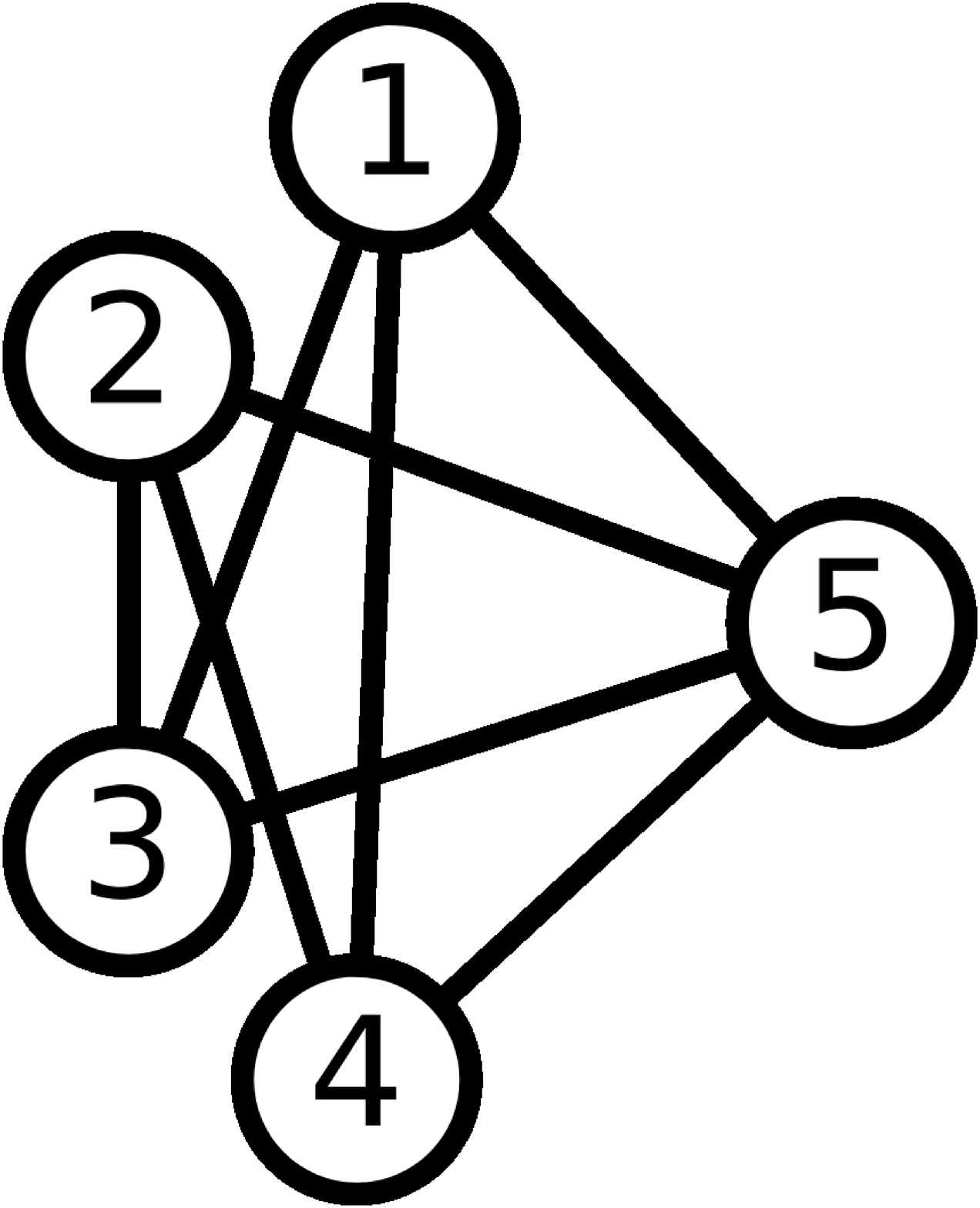}
\begin{align*}
C_1&=\{1,2\},\\ C_2&=\{3,4\},\\ C_3&=\{5\}.
\end{align*}
(c)
\end{minipage}
\caption{Examples of 2-partite (a)-(b) and 3-partite  (c) graphs.}
\label{fig:partite}
\end{figure}

\section{Infrastructure-based networks}
\label{sec:downlink}
We have so far considered a network in ad-hoc mode, without infrastructure. We
now consider $N$ access points to which users must connect. In particular, each
class now corresponds either to uplink traffic (from the users to an access point)
or  to downlink traffic (from an access point to the users). We study  the
flow-level dynamics of CSMA under the time-scale separation assumption.
Specifically, we prove the suboptimality of standard CSMA in this context and
introduce a slight modification of CSMA, we refer to as \textit{flow-aware} CSMA,
which makes the algorithm optimal.

\subsection{Uplink vs.\,downlink}

For all $i=1,\ldots,N$, we denote by $U_i$ and $D_i$ the sets of  uplink and
downlink classes, respectively, associated with access point $i$. In the example of figure \ref{fig:down}, for instance, there are $N=2$ access points and  $K=6$ classes, with  $U_1=\{2\}$, $D_1=\{1,3\}$, $U_2=\{5\}$ and $D_2=\{4,6\}$.
An access point cannot transmit and receive on the same channel. In particular, 
those classes sharing the same access point, either in uplink or downlink, conflict
with each other. Formally, for all access points $i=1,\ldots,N$ and all classes
$k,l\in U_i\cup D_i$, we have $(k,l)\in E_j$ for each channel $j$ such that
$k,l\in V_j$. We assume that an access point cannot transmit data on more than
one channel at a time but is able to receive data on the $J$ channels simultaneously.

\begin{figure}[h]
\centering
\begin{minipage}[c]{0.6\linewidth}
\includegraphics[width=\linewidth]{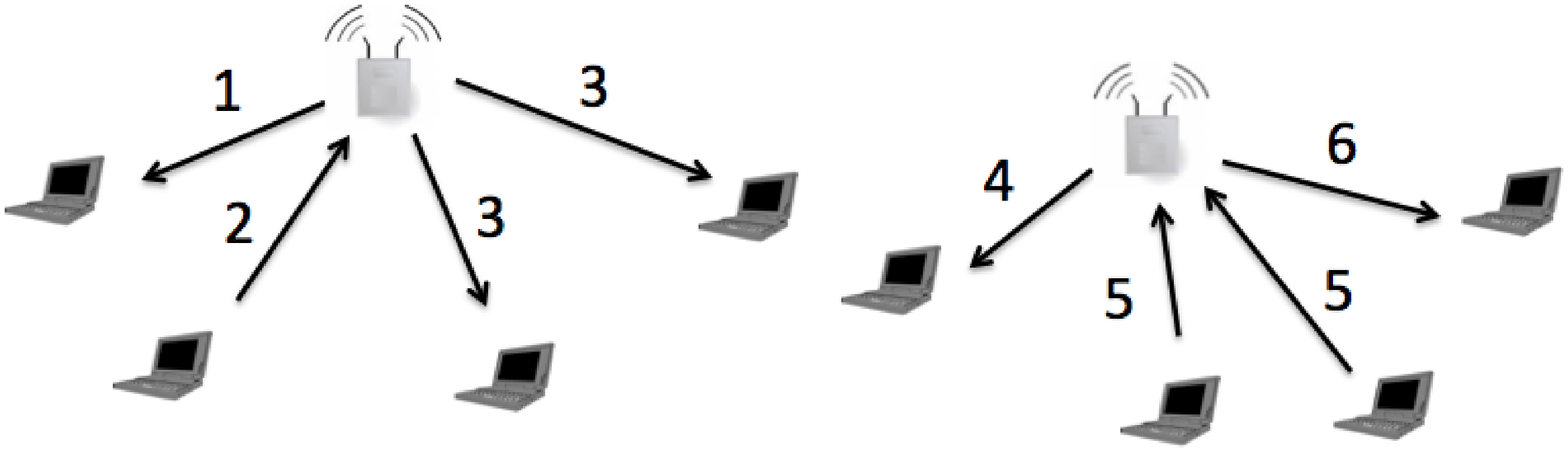}
\end{minipage}
\hspace{0.05\linewidth}
\begin{minipage}[c]{0.25\linewidth}
\includegraphics[width=\linewidth]{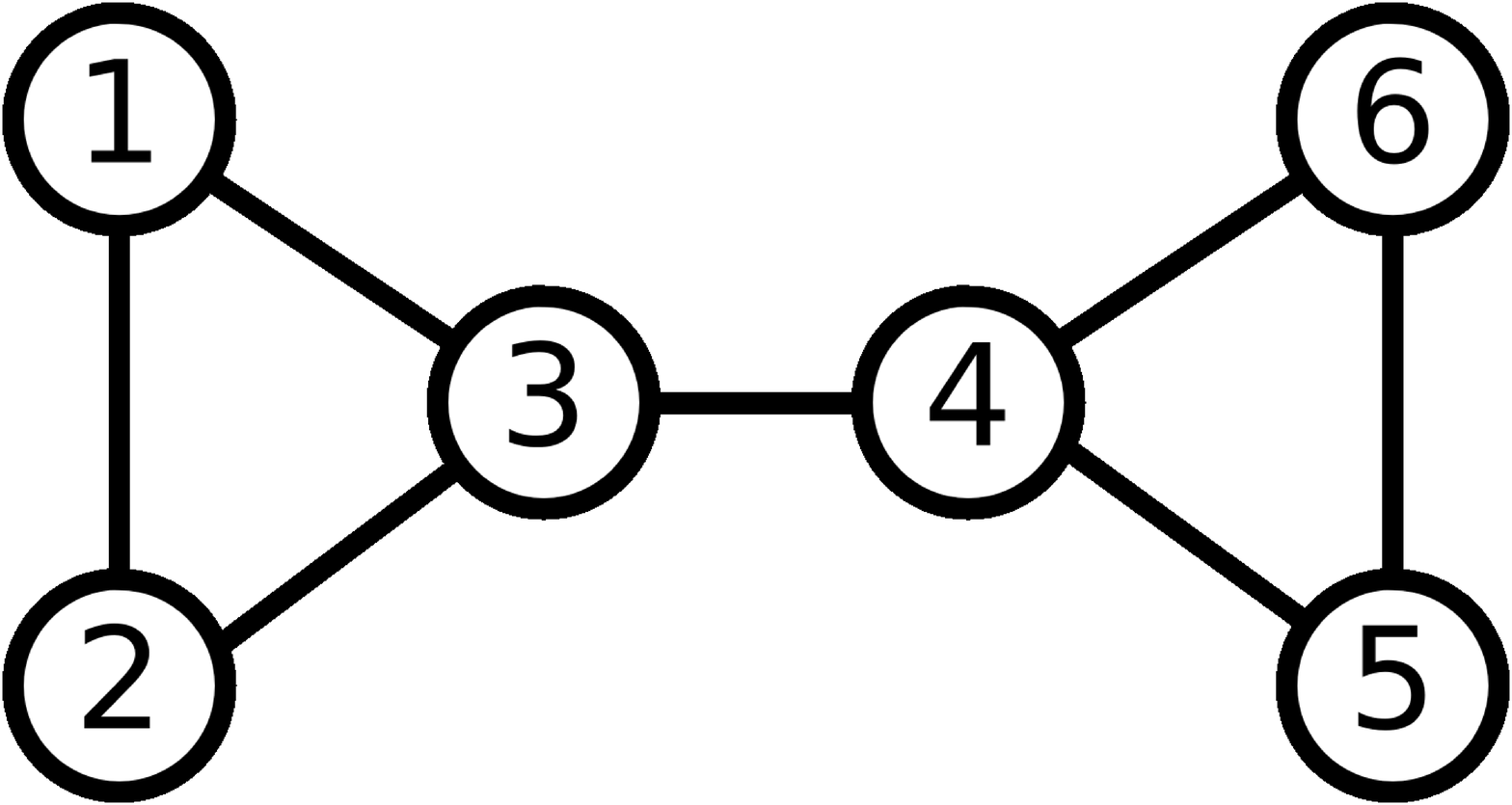}
\end{minipage}
\caption{A network of 2 access points with 6 classes of links and its interference graph.}
\label{fig:down}
\end{figure}

The feasible schedules are those defined in section \ref{sec:fs}, with the
additional constraint that each access point cannot transmit data on more than one
channel at a time, that is:
\begin{equation}\label{eq:yd}
\forall i=1,\ldots,N, \quad \sum_{k\in D_i} y_k\le 1.
\end{equation}
We denote by  $\Y(x)$ the set of feasible schedules and by $\Y$ the  union of $\Y(x)$
over all network states $x$.  The corresponding capacity region is defined in
section \ref{capa}.

\subsection{Standard CSMA}

We first consider the standard CSMA algorithm: each transmitter  waits for a period
of random duration before attempting transmission on some randomly chosen channel. The key difference with the 
ad-hoc wireless network considered so far is that each  access point runs a single instance of the CSMA algorithm for all its downlink traffic. In particular, 
for each access point $i$, the attempt rates $\nu_k$ are the same for all classes $k\in D_i$. At each attempt, the 
access point $i$ selects a class-$k$ flow with some probability proportional to $x_k$ and probes channel $j$ with probability  $\beta_{kj}$.
If the probed channel is sensed idle, a packet of this flow is transmitted.

It is worth noting that the attempt rate of each access point is independent of its congestion level, in terms of the number of ongoing downlink  flows at this access point. This breaks the natural stabilizing effect of CSMA we have proven in Theorem \ref{optim} in the context of ad-hoc networks, where those classes  with a higher number of flows get preferential access to the radio channels.
In the following, we illustrate the suboptimality of standard CSMA on two examples with downlink traffic only. 
Note that, in the presence of uplink traffic only, the model is in fact equivalent to the ad-hoc network considered so far. 

For this purpose, we give the distribution of feasible schedules achieved by the
algorithm under the time-scale separation assumption. Denoting by $Y(t)$ the
schedule at time $t$, we have the analogue of Proposition \ref{prop1}:

\begin{proposition}
If both the packet sizes and the
backoff times  have exponential distributions, then $Y(t)$  is a reversible Markov process, with stationary measure:
\begin{equation}\label{eq:wix2}
\begin{aligned}
w(x,y)=& \prod_{i=1}^N \prod_{k\in U_i:x_k>0} {x_k!\over (x_k-y_k)!}\alpha_k^{y_k}\prodj {\beta_{kj}^{y_{kj}}}\\
&\times
{\left( \sum_{k\in D_i} x_k\right)!}\prod_{k\in D_i:x_k>0} {\alpha_k^{y_k}\over x_k!} \prodj {\beta_{kj}^{y_{kj}}}
,\quad y\in \Y(x).
\end{aligned}
\end{equation}
\end{proposition}
\begin{proof}
As for Proposition \ref{prop1}, the proof follows from the local balance equations. For all $i,\ldots,N$, we have:
$$
\forall k\in U_i,\quad w(x,y) (x_k-y_k)\nu_k\beta_{kj}=w(x,y+e_{kj}){\varphi_k},
$$
and
$$
\forall k\in D_i,\quad w(x,y) {x_k\over \sum_{k\in D_i}x_k}\nu_k\beta_{kj}=w(x,y+e_{kj}){\varphi_k}.
$$
\ep
\end{proof}

The stationary distribution of the schedules $\pi(x,y)$ follows from normalization.
Again, it is insensitive to the  packet size and backoff time distributions beyond
the means. The throughput of class $k$ is given by \eqref{eq:phix}.

\paragraph{Example 1}
The most simple example showing the suboptimality of CSMA is shown in Figure
\ref{fig:ref}. It consists of $N=3$ access points, a single class per access point
and a single channel.  Taking unit physical rates, the optimal stability region is
$\rho_1+\rho_2<1$ and $\rho_2+\rho_3<1$ where 1 and 3 are the edge classes and 2
is the center class. We have proven in \cite{BF-10} that the actual stability
region is strictly smaller, even in the limiting case of infinite attempt rates.

\begin{figure}[h]
\centering
\begin{minipage}[c]{0.6\linewidth}
\includegraphics[width=\linewidth]{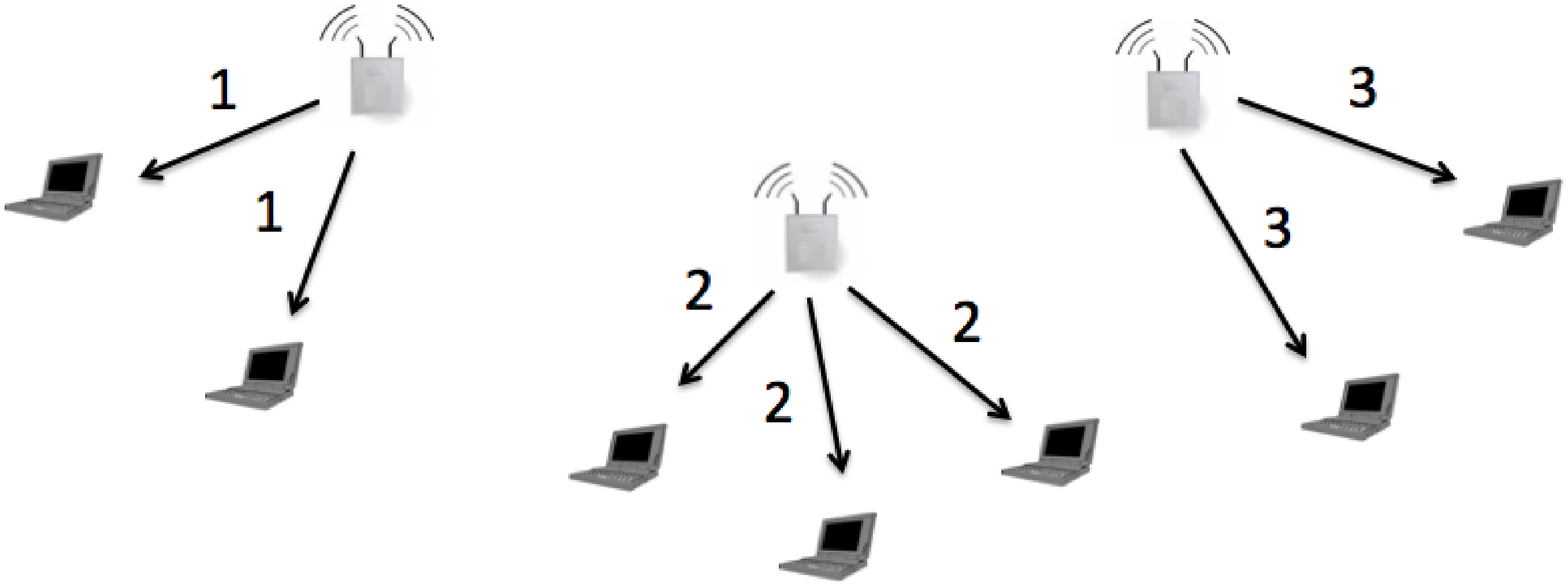}
\end{minipage}
\hspace{0.05\linewidth}
\begin{minipage}[c]{0.25\linewidth}
\includegraphics[width=\linewidth]{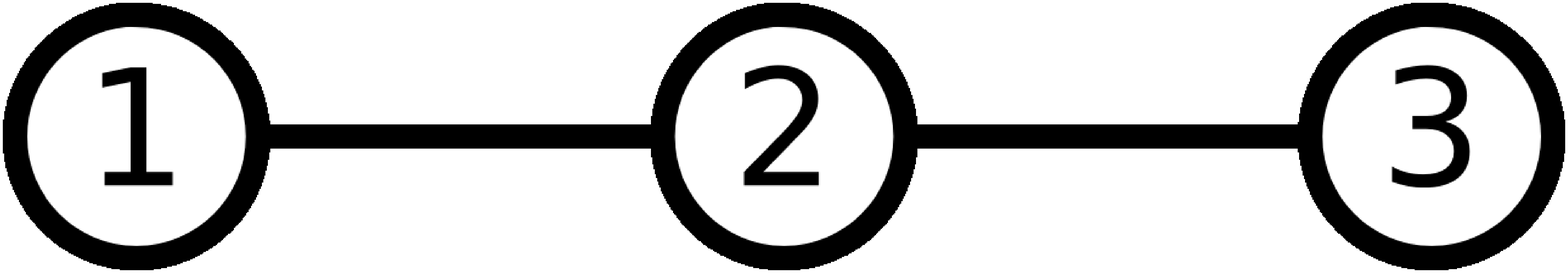}
\end{minipage}
\caption{Network of 3 access points with a single downlink class per access point and its interference graph.}
\label{fig:ref}
\end{figure}

\paragraph{Example 2}
Consider the multi-channel network of Figure \ref{fig:bow}  with $N=5$
access points, a single class per access point and $J=2$ channels, further referred to as the \textit{bow tie network}.
The conflict graph is the same for both channels.
We refer to class 3 as the center class and to the other classes as the edge classes.
We assume that the mean packet sizes and the  mean backoff times are the same for all classes, so that $\alpha_k=\alpha$ for all $k=1,\ldots,5$, for some $\alpha>0$.
We also assume that all classes except class 3 have the same traffic intensities.
The optimal stability condition is then given by:
\begin{equation}\label{eq:opt-tie}
\rho_3<1 \quad  \mathrm{and}\quad 2\rho_1+\rho_3<2.
\end{equation}

\begin{figure}[ht]
\centering
\begin{minipage}[c]{0.6\linewidth}
\includegraphics[width=\linewidth]{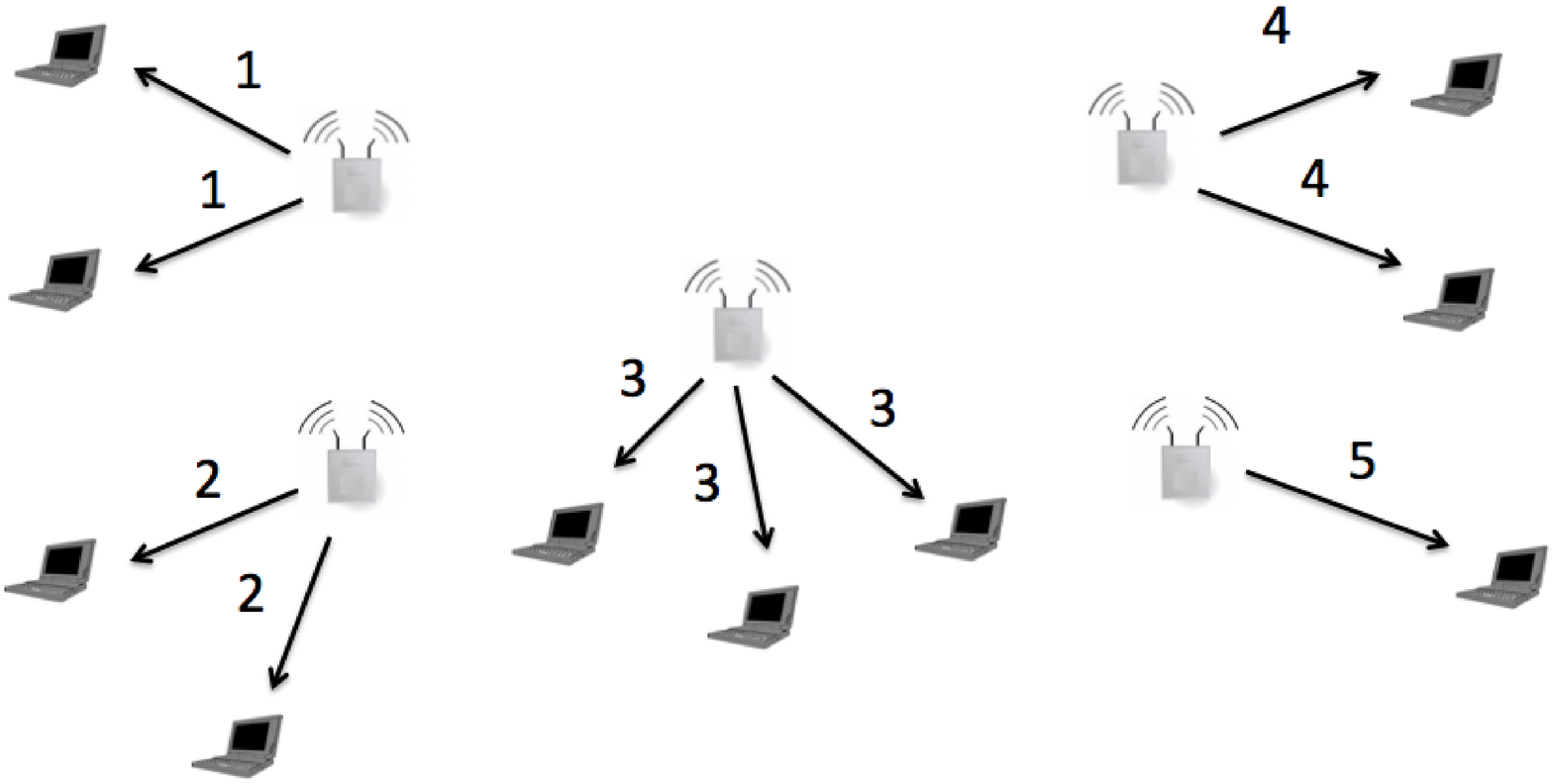}
\end{minipage}
\hspace{0.05\linewidth}
\begin{minipage}[c]{0.25\linewidth}
\includegraphics[width=\linewidth]{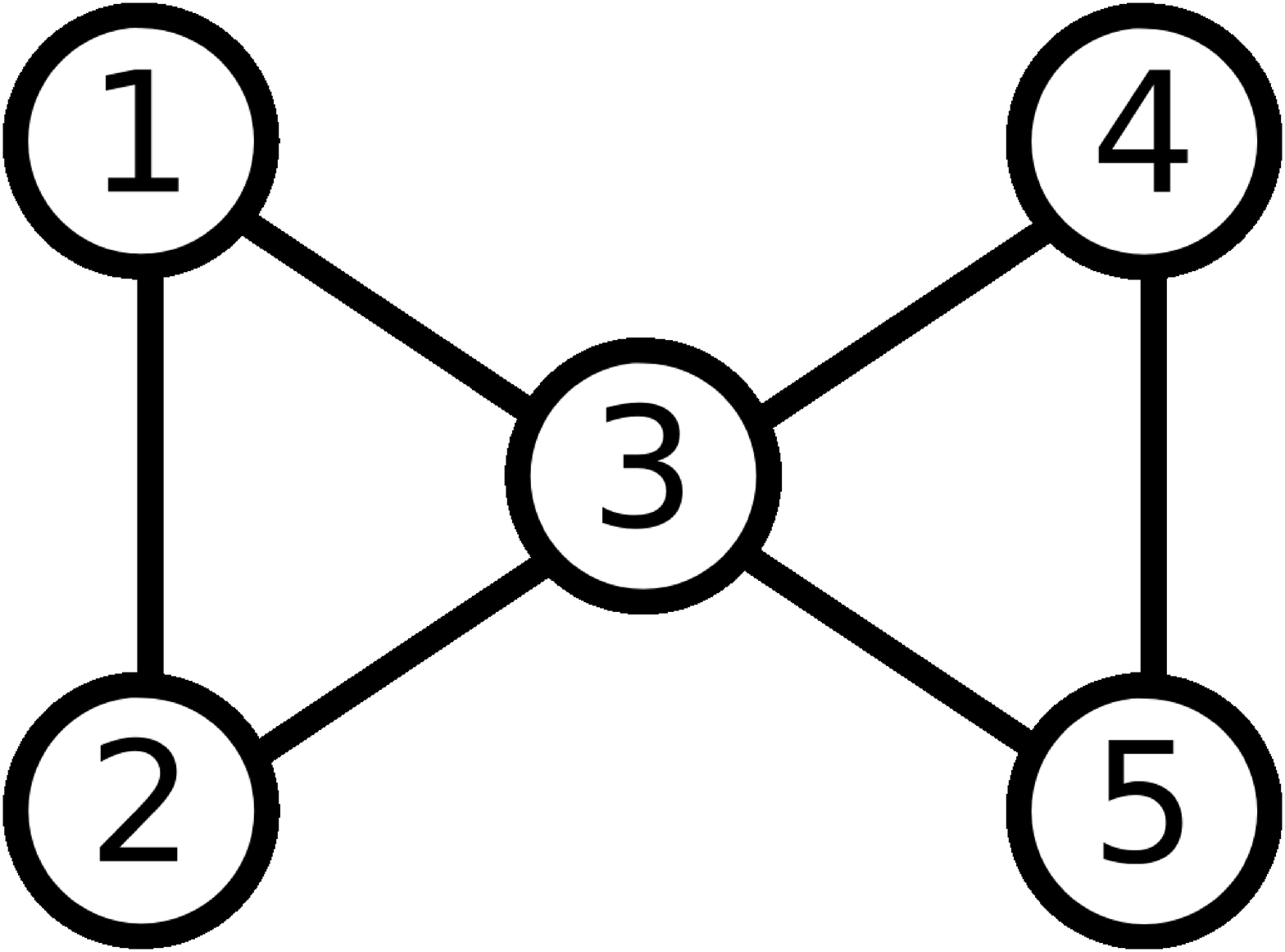}

\end{minipage}
\caption{Network of 5 access points with a single downlink class per access point and its interference graph.}
\label{fig:bow}
\end{figure}

We consider the limiting case where $\alpha\rightarrow\infty$ and we assume that
the two channels are chosen uniformly at random. We then deduce from \eqref{eq:wix}-\eqref{eq:phix} the following throughput vector:
\begin{equation}\label{eq:mu5}
\phi(x)=\left\{
\begin{array}{lll}
(1,1,0,1,1)&\text{if }&x_1, x_2, x_4, x_5>0,\\
\left(\frac{3}{4},\frac{3}{4},\frac{1}{2},1,0\right)&\text{if }&x_1, x_2, x_3, x_4>0, x_5=0,\\
\left(\frac{2}{3},\frac{2}{3},\frac{2}{3},0,0\right)&\text{if }& x_1, x_2, x_3>0, x_4=x_5=0,\\
(0,1,1,1,0)&\text{if }&x_2, x_3, x_4>0, x_1=x_5=0,\\
(1,1,0,0,0)&\text{if }&x_1,x_2>0, x_3=x_4=x_5=0,\\
(1,0,0,0,0)&\text{if }&x_1>0,x_2=0, x_3=x_4=x_5=0.
\end{array}
\right.
\end{equation}
The other cases follow by symmetry.
The center class is in conflict with all other classes for accessing
the channels and  is either not served when the 4  other classes are active or
served at a low rate when 3 other classes are active.
This also  results in a suboptimal stability region:

\begin{proposition}
The bow tie network is unstable whenever:
\begin{equation}\label{eq:subopt-tie}
\rho_3>\frac{1}{3}\rho_1^4 - \frac{2}{3}\rho_1^3 - \frac{2}{3} \rho_1^2 +1.
\end{equation}
\label{prop:suboptimal-2}
\end{proposition}

This proposition is proven in the appendix.
In the homogeneous case $\rho_1=\rho_3$ for instance, Proposition \ref{prop:suboptimal-2} implies that the network is unstable whenever $\rho_1>0.63$.
In view of \eqref{eq:opt-tie}, the optimal stability condition is $\rho_1<2/3$, which
shows that the standard CSMA algorithm is not optimal.  This suboptimality is illustrated by Fig.~\ref{fig:suboptimal},
the actual stability condition being obtained by the simulation of the underlying Markov process.
 In the homogeneous case   for instance, the  loss of efficiency is around 15\%.

\begin{figure}[ht]
\begin{center}
\includegraphics[width=.4\linewidth,angle=-90]{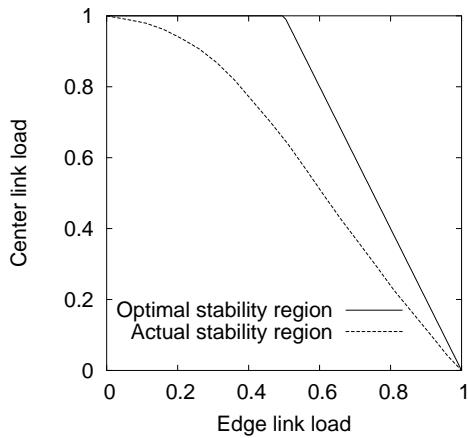}
\caption{Stability region of the bow-tie network with two channels under standard CSMA.}
\label{fig:suboptimal}
\end{center}
\end{figure}

\subsection{Flow-aware CSMA}

The flow-aware CSMA algorithm consists for each access point to run one standard
CSMA algorithm per flow. This compensates for the inherent bias of standard
CSMA against downlink flows and stabilizes the network whenever possible. Indeed,
the stationary measure of the schedules is now given by \eqref{eq:wix}. The only
difference with the ad-hoc wireless network considered in section \ref{sec:flow-level}
is the additional constraint \eqref{eq:yd} on the set of feasible schedules.
This does not change the proof of Theorem \ref{optim}, showing the optimality of
flow-aware CSMA.

\section{Conclusion}
\label{sec:conc}
We have proved that, under the time-scale separation assumption, the distributed scheduling achieved by standard CSMA  exploits the radio
resources in an optimal way  in  ad-hoc wireless networks. This is not the case in the presence of access points,
due to the inherent bias of CSMA against downlink traffic.  A slight modification of CSMA we refer to as flow-aware CSMA is then sufficient to correct this bias and to make the 
algorithm optimal. 

The analysis relies on a number of simplifying assumptions that we plan to relax
in future work. First,  we have neglected the impact of packet collisions; these
could be included in the model, as done in \cite{walrand-09} for rate-based adaptive CSMA
for instance. One may then account for the adaptive backoff of the IEEE 802.11 protocol, which is key in practice to limit the number of  collisions. Other issues that may be worth  addressing concern the traffic model.
We have neglected the impact of acknowledgements, which are known to be critical in IEEE 802.11 networks. 
The impact of real-time traffic
should also be considered. Finally, one may think of multi-hop networks where the
flows of some source-destination pairs must go through one or several relay
nodes. Although we believe that flow-aware CSMA is still optimal in this more
general settings, we have not yet been able to prove this result. 

From a more theoretical perspective, one may relax the assumption of Poisson flow arrivals and
exponential flow sizes in the stability analysis. One may for instance consider user sessions that
consist of an alterning series of file transfers and idle periods.
We would also like to extend Proposition \ref{bipart} to any interference graph,  which would prove the validity of Theorem \ref{optim} in the absence of the time-scale separation assumption.

\appendix

\section*{Appendix}

\paragraph{Proof of Lemma 1}
For any  class $k$, let:
$$
\beta_k=\min_{j:k\in V_j}\beta_{kj}.
$$
Note that $\beta_k>0$.
We have for all $y\in \Y(x)$:
$$
w(x,y)\ge  \prod_{k:x_k>0} {x_k(x_k-1)\ldots  (x_k-y_k+1)\over x_k^{y_k}}{\beta_k^J} u(x,y).
$$
If $x_k\le 2J$, we have:
$$
{x_k(x_k-1)\ldots  (x_k-y_k+1)\over x_k^{y_k}}\ge {1\over x_k^{y_k}} \ge {1\over (2J)^J}.
$$
Otherwise, we have using the fact that $y_k\le J$ for all $k=1,\ldots,K$:
$$
{x_k(x_k-1)\ldots  (x_k-y_k+1)\over x_k^{y_k}}\ge \left({x_k-y_k+1\over x_k}\right)^{y_k} \ge {1\over 2^J}.
$$
Combining these results, we obtain the existence of some constant $m>0$ such that:
$$
\forall y\in \Y(x),\quad w(x,y)\ge  m u(x,y).
$$

Now let:
 $${\cal Z}(x)=\left\{y\in \Y(x): \log(u(x,y))\ge  (1-{\epsilon\over 2}) \log(u(x))\right\}.$$
We have:
$$
\sum_{y\in \Z(x)} \pi(x,y)\log(u(x,y))\ge (1-{\epsilon\over 2}) \log(u(x))\sum_{y\in {\cal Z}(x)} \pi(x,y).
$$
Using the fact that $w(x,y)\le u(x,y)$ for all $y\in \Y(x)$, we get:
\begin{align*}
\sum_{y \in \Y(x)\setminus \Z(x)} \pi(x,y)&=  {\sum_{y \in \Y(x)\setminus \Z(x)} w(x,y)\over \sum_{y\in \Y(x)} w(x,y)}, \\
&\le   {1\over m}  {\sum_{y\in \Y(x)\setminus \Z(x)} u(x,y)\over \sum_{y\in \Y(x)} u(x,y)}, \\
&\le   {1\over m} {M u(x)^{1-{\epsilon\over 2}}\over  \max_{y\in\Y(x)}u(x,y)}, \\
&=  {1\over m} {M \over  u(x)^{\epsilon\over 2}},
\end{align*}
where $M$ denotes the total number of schedules (that is, the cardinal of $\Y$).
Since $u(x)$ tends to $+\infty$ when $|x|\equiv \sum_k x_k$ tends to  $+\infty$, this quantity
is less than $\epsilon/2$ for all states $x$ but some finite number. In those states, we have:
$$
\sum_{y\in {\cal Z}(x)} \pi(x,y)\ge 1-{\epsilon\over 2}.
$$
We deduce that in all states $x$ but some finite number:
\begin{align*}
\sum_{y\in\Y(x)}\pi(x,y)\log(u(x,y))&\ge (1-{\epsilon\over 2})^2 \log(u(x)),\\
&\ge (1-{\epsilon}) \log(u(x)).
\end{align*}
\ep

\paragraph{Proof of Lemma 2}

Let:
$$
v(x,y)= \prod_{k:x_k\ge J}(x_k{\alpha_k})^{y_k}.
$$
There are some positive constants $m,M$ such that:
$$
\forall x\in \N^K,\ \forall y\in \Y,\quad m\le {u(x,y)\over v(x,y)}\le M.
$$
The proof then follows from the fact that:
$$
v(x)=\max_{y\in \Y}u(x,y)
\le M\max_{y\in \Y}v(x,y)
= M\max_{y\in \Y(x)}v(x,y)
\le {M\over m}\max_{y\in \Y(x)} u(x,y)
= {M\over m} u(x).
$$
\ep

\paragraph{Proof of Theorem 1}
\label{proof:th1}

If the vector of traffic intensities lies in the interior of the capacity region, there exist some $\epsilon>0$ and some probability measure $\pi$ on $\Y$  such that:
\begin{equation}\label{eq:stab}
\forall k=1,\ldots,K,\quad\rho_k=\varphi_k(1-2\epsilon)\sum_{y\in \Y}\pi(y)y_k.
\end{equation}
Note that we can choose $\pi(y)>0$ for all $y\in \Y$.

Define the Lyapunov function:
$$
F(x)=\sum_{k:x_k>0} {x_k\sigma_k\over\varphi_k} \log(x_k\alpha_k).
$$
The corresponding drift  is given by:
\begin{align*}
\Delta F(x)=&\sum_{k} \lambda_k (F(x+e_k)-F(x))
+ \sum_{k:x_k>0} {\phi_k(x)\over \sigma_k} (F(x-e_k)-F(x)),
\\
=&\sum_{k:x_k=0} {\rho_k\over \varphi_k} \log(\alpha_k) +\sum_{k:x_k>0} {\rho_k\over \varphi_k} \left((x_k+1)\log((x_k+1)\alpha_k) -x_k\log(x_k\alpha_k)\right)\\
&+ \sum_{k:x_k>0} {\phi_k(x)\over\varphi_k} \left((x_k-1)\log((x_k-1)\alpha_k)-x_k\log(x_k\alpha_k)\right).
\end{align*}
In particular, we have $\Delta F(x)=G(x)+H(x)$ with:
$$
G(x)=\sum_{k:x_k>0}{\rho_k-\phi_k(x) \over \varphi_k}\log(x_k\alpha_k),
$$
$$
H(x)=
\sum_{k:x_k>0}{\rho_k\over \varphi_k}(x_k+1)\log(1+{1\over x_k})
+\sum_{k:x_k>0}{\phi_k(x) \over \varphi_k}(x_k-1)\log(1-{1\over x_k})+\sum_{k:x_k=0} {\rho_k \over \varphi_k}\log(\alpha_k),
$$
where we use  the convention $0\log(0)\equiv 0$.
Since $\phi_k(x)\le J\varphi_k $, the function $H(x)$ is bounded.
Regarding $G(x)$, it follows from
 (\ref{eq:weight}) and (\ref{eq:stab}) that:
\begin{align*}
G(x)&= \sum_{y\in\Y}((1-2\varepsilon)\pi(y)-\ \pi(x,y))\sum_{k:x_k>0}y_k
  \log(x_k\alpha_k),\\
  &= \sum_{y\in \Y}((1-2\varepsilon)\pi(y) -\pi(x,y))\log(u(x,y)).
\end{align*}
By Lemma 1, we have for all states $x$ but some finite number:
\begin{align*}
G(x)&\le -{\epsilon} \sum_{y\in\Y} \pi(y) \log(u(x,y))+ (1-{\epsilon})\left(\sum_{y\in\Y} \pi(y) \log(u(x,y))-\log(u(x))\right),\\
& \le -{\epsilon} \sum_{y\in\Y} \pi(y) \log(u(x,y))+ (1-{\epsilon})\log\left({v(x)\over u(x)}\right).
\end{align*}
Since $\pi(y)>0$ for all $y\in\Y$, the first term tends to $-\infty$ when  $|x|\equiv \sum_k x_k$ tends to  $+\infty$. By Lemma 2, the second term is bounded. We deduce the existence of some $\delta>0$ such that $\Delta F(x)\le -\delta$ for all states $x$ but some finite number. The proof then
 follows from  Foster's criterion.
 \ep

\paragraph{Proof of Theorem 2}
In the following, we consider $(X^N(t))_{N\ge 1}$ as a sequence of stochastic processes
in the space $\D_{\N^{K}}([0,\infty[)$ of \textit{c\`ad-l\`ag} functions with values in $\N^K$ with the Skorohod topology.

First, we have to prove the tightness of the sequence $(X^N(t))$. It is enough
to remark that, for all $N\ge 1$, $X^N_k(t)$ is stochastically dominated by a Poisson
process of intensity  $\lambda_k$ and stochastically dominates an $M/M/1$ queue with arrival rate
$\lambda_k$ and service rate $\varphi_k/\sigma_k$. Thus, the conditions of the Arzel\`a-Ascoli theorem are
fulfilled and the sequence $(X^N(t))$ is tight  (see \cite[Th 12.3]{billingsley-99}).

We now consider a bounded function $f$ on $\N^K$. Denote by $\Omega^N$ the infinitesimal generator of the
Markov process $(X^N(t),Y^N(t))$. For $x\in\N^K$ and $y\in\Y$, we have
$$
\Omega^N(f)(x,y) = \sumk \lambda_k (f(x+e_k)-f(x)) - \sumk \varphi_k/\sigma_k  \sumj y_{kj} (f(x-e_k)-f(x)).
$$
 According to the Martingale characterization
of Markov jump processes (see \cite{rogers-00}), the process:
$$
M_f^N(t) = f(X^N(t)) - f(X^N(0)) - \int_0^t \Omega^N(f)(X^N(s),Y^N(s))\diff s
$$
is a locale martingale and, since the process $X^N(t)$ is not exploding on $[0,t]$
(it is stochastically dominated by a Poisson process), it is a martingale.

For each $N\ge 1$, define the random measure:
$$
\Gamma^N([0,t]\times B) = \int_0^t \ind_{\{Y(s)\in B\}} \diff s,\quad \text{for } B\subset \Y
$$
$\Gamma^N$ is a random variable with value in the set $\mathcal{L}(\Y)$ of the random measures
on $[0,\infty[\times\Y$ such that if $\mu\in\mathcal{L}(\Y)$ then $\mu([0,t]\times\Y) = t$
for all $t\geq 0$. Since $\Y$ is finite, the set $\mathcal{L}(\Y)$ is compact
and then the sequence $(\Gamma^N)_{N\ge 1}$ is relatively compact.

Assume that the sequence $(X^N(t),\Gamma^N)_{N\ge 1}$ tends to  some limit $(X(t),\Gamma)$.
Since:
$$
\int_0^t \Omega^N(f)(X^N(s),Y^N(s))\,\diff s = \int_0^t \sumy \Omega^N(f)(X^N(s),y) \Gamma^N(\diff s\times \diff y)
$$
and $f$ is bounded,  this random variable tends in distribution to:
$$
\int_0^t \sumy \Omega(f)(X(s),y) \Gamma(\diff s\times \diff y).
$$
It remains to  characterize $\Gamma$. According to Lemma 1.3 of \cite{kurtz-92},
there exists a set of random probability measures $\vartheta(t,.)$ on $\Y$ such that:
$$
\Gamma([0,t]\times B)= \int_0^t \vartheta(s,B)\diff s ,\quad \text{for } B\subset \Y.
$$
For any function $g$ on $\Y$, we define the martingale:
$$
\bar{M}_g^N(t) = \frac{1}{N}\left(g(Y^N(t)) - g(Y^N(0)) - \int_0^t \Omega^N(g)(X^N(s),Y^N(s))\diff s \right).
$$
For $x\in\N^K$ and $y\in\Y$, we have:
\begin{align*}
\Omega^N(g)(x,y) = \sumk \sumj &N(x_k-y_k)\nu_k \beta_{kj} (g(y+e_{kj})-g(y))\\
& + \left( Ny_{kj} \varphi_k\left(1 - \frac{1}{\sigma_kN}\right) + \frac{\varphi_k}{\sigma_k}\right) (g(y-e_{kj})-g(y)).
\end{align*}
The increasing process of this martingale is:
\begin{align*}
\left\langle\bar{M}_g^N(t)\right\rangle &= \frac{1}{N^2}\int_0^t \Omega^N(g)(X^N(s),Y^N(s))\diff s,\\
&\leq \frac{2t }N \max_{y\in\Y}|g(y)| \left(\max_{k}\varphi_k + \max_{k}\nu_k \max_{k,j} \beta_{kj}\right).
\end{align*}
It tends to 0 on all compact sets so that the martingale tends in distribution to 0.
Since $\Y$ is finite, $g$ is bounded and $(g(Y^N(t)) - g(Y^N(0)))/N$ also tends to 0.
Finally, we get that: $${1\over N}\int_0^t \Omega^N(g)(X^N(s),Y^N(s))\diff s$$
converges in distribution to 0. This implies:
\begin{align*}
\int_0^t \sumy \Biggl( \sumk \sumj  &(X_k(s)-y_k)\nu_k \beta_{kj} (g(y+e_{kj})-g(y)) \\
&+ y_{kj} \varphi_k (g(y-e_{kj})-g(y))\Biggr) \vartheta(s,y) \diff s = 0
\end{align*}
and for almost every $s$ in $[0,t]$, we have:
\begin{align*}
\sumy \Biggl( \sumk \sumj  &(X_k(s)-y_k)\nu_k \beta_{kj} (g(y+e_{kj})-g(y))\\
&+ y_{kj} \varphi_k(g(y-e_{kj})-g(y))\Biggr) \vartheta(s,y) =0
\end{align*}
The probability distribution $\vartheta(s,.)$ is then the stationary distribution
given by \eqref{eq:wix}.

It follows that:
$$
\int_0^t \Omega^N(f)(X^N(s),Y^N(s))\,\diff s
$$
converges in distribution to:
$$
\int_0^t \Omega(f)(X(s))\,\diff s
$$
where $\Omega$ is the infinitesimal generator of the Markov process described
in section \ref{sec:flow-level}. For $x\in\N^K$, we have
$$
\Omega(f)(x) = \sumk \lambda_k \left(f(x+e_k)-f(x)\right) + \phi_k(x)/\sigma_k \left(f(x-e_k)-f(x)\right),
$$
where $\phi_k(x)$ is the mean throughput of class $k$ in state $x$, given by
 \eqref{eq:phix}.

By  dominate convergence, $M_f^N(t)$ tends in distribution to:
$$
M_f(t) = f(X(t)) - f(X(0)) - \int_0^t \Omega(f)(X(s))\diff s,
$$
and $M_f(t)$ is a martingale. Using the characterization of the Markov jump processes, we get that
the process $X(t)$ is a Markov process with infinitesimal generator $\Omega$.

This concludes the proof.
\ep

\paragraph{Proof of Proposition 2}
For this proof, we will need the notion of fluid limits. A fluid
limit  is a limiting point $\bar X^N(t)$ of the laws of the processes $\{X^N(n t)/n, n\ge 1\}$
in the set of probability measures on the space $\D_{\R_+^K}([0,\infty))$ of \textit{cad-lag}
functions with value in $\R_+^K$ with Skorohod topology (see \cite{billingsley-99}).
It is not difficult to show that the set of processes $\{X^N(n t)/n, n\ge 1\}$
is tight in the set of probability distributions on the
space $\D_{\R_+^K}([0,\infty))$ endowed with the metric associated to the
uniform norm on compact sets. Therefore, there exists at least one fluid limit
and any fluid limit is continuous.
Since the process $Y^N(n t)$ has its values in a finite
space for all $n\ge 1$, it can be proved as in \cite{dai-95,robert-03} that, if there exists a deterministic time $T>0$
such that $\bar X^N(t)=0$  for all $t\geq T$,  then
the Markov process $(X^N(t),Y^N(t))$ is ergodic. 

The  proof is then very similar to that given in \cite{proutiere} for random capture algorithms.
We consider a fluid limit $\bar X^N(t)$ and define:
$$
W^N(t) = \sum_{i=1}^L \max_{k \in C_l} \left(\bar X^N_k(t){\sigma_k\over \varphi_k}\right).
$$
When some  class in $C_l$ takes channel
$j$, all other classes in $C_l$ can take this channel while all classes in $\{1,\ldots,K\}\setminus C_l$ cannot.
 This implies:
$$
W^N(t) \leq \max\left(0, 1 + \left(\sum_{l=1}^L \max_{k \in C_l} \frac{\rho_k}{\varphi_k} -J\right)t\right).
$$
In the case of $L$-partite networks, the capacity region is given by the set of vectors $\phi$ such that:
$$
 \sum_{l=1}^L \max_{k \in C_l} \frac{\phi_k}{\varphi_k} \leq J.
$$
Since $\rho$ lies inside the capacity region, we have $W^N(t)=0$  for all $t\geq T$, with 
$$
T= {1\over J-\sum_{i=1}^L \max_{k \in C_l} \rho_k},
$$
which implies the ergodicity of the Markov process $(X^N(t),Y^N(t))$.
\ep

\paragraph{Proof of Proposition \ref{prop:suboptimal-2}}
Define the throughput vector $\tilde{\phi}$ such that $\tilde{\phi}_3(x)=\phi_3(x)$
 and $\tilde{\phi}_k(x)=\ind_{\{x_k>0\}}$ for all $k\ne 3$.
 It can be easily verified that  $\tilde{\phi}_k(x) \geq \tilde{\phi}_k(y)$ for all states
$x,y$  such that $x\le y$ and all $k$ such that $x_k>0$. Now
consider the coupling of the stochastic processes $X(t)$ and $\tilde{X}(t)$ describing the evolution of the queues
for the throughputs $\phi$ and $\tilde{\phi}$, respectively, starting from the same initial state
$X(0)=\tilde{X}(0)$.
It follows from the above monotonicity property that $\tilde{X}(t) \le X(t) $ a.s.~at any time $t\ge 0$.
In particular, the transience or the null recurrence of $\tilde X(t)$ implies that
of ${X}(t)$.

For the throughput vector $\tilde \phi(t)$, queues 1,2,4,5 are  independent
$M/M/1$ queues with load $\rho_1$. If $\rho_1 \geq 1$, the Markov process $\tilde X(t)$ is null recurrent or transient.
Note that \eqref{eq:subopt-tie} then reduces to $\rho_3\geq 0$.

Assume now that $\rho_1<1$. To prove the transience of
$\tilde X(t)$, we use fluid limits.
Since $\rho_1<1$ and for $\tilde \phi(t)$, queues 1,2,4,5 are  independent
$M/M/1$ queues with load $\rho_1$, there exists some
finite time after which, for any initial conditions,
the corresponding components of the fluid limit are null. We then just have to
consider the fluid limits with the
initial condition $\bar X_3(0)=1$ and $\bar X_k(0)=0$ for  all $k\neq 3$. In this case,
Proposition 9.14 of \cite[p.241]{robert-03} applies and the fluid limit satisfies:
$$
 \bar X_3(t)= 1+ \left( \lambda_3 - \frac{\bar \phi_3}{\sigma_3} \right)t,$$
  as long as this function is positive, where $\bar \phi_3$ is the throughput of link 3 averaged over the states of other links.
  Since each other link is active with probability $\rho_1$, it follows from \eqref{eq:mu5} that:
  $$
  \bar\phi_3=\frac{1}{3}\rho_1^4 - \frac{2}{3}\rho_1^3 - \frac{2}{3} \rho_1^2+1.
  $$
In particular,  $\bar{X}_3(t)$ increases linearly to infinity whenever inequality \eqref{eq:subopt-tie} is satisfied and, according to \cite{meyn-95},
the Markov process $\tilde{X}(t)$ is transient.
\ep

\bibliographystyle{plain}
\bibliography{biblio.bib}

\end{document}